\documentclass[conference]{IEEEtran}

\IEEEpubid{\makebox[\columnwidth]{\hfill 9781-4244-3941-6/09/\$25.00~\copyright~2009 IEEE}
\hspace{\columnsep}\makebox[\columnwidth]{Published by the IEEE Computer Society}}

\usepackage{url,color}
\usepackage[pdftex]{graphicx}
\usepackage[cmex10]{amsmath}
\usepackage{amssymb}
\usepackage{amsthm}
\usepackage{float}
\usepackage{epstopdf}
\usepackage{hyperref}

\newtheorem{p}{Proposition}

\ifCLASSINFOpdf
\fi
\IEEEoverridecommandlockouts
\begin{document}

\title{Proactive Location-Based Scheduling of Delay-Constrained Traffic Over Fading Channels}

\author{\thanks{This publication was made possible by NPRP grant $\#$~NPRP $7-923-2-344$ from the Qatar National Research Fund (a member of Qatar Foundations). The statements made herein are solely the responsibility of the authors.
}\IEEEauthorblockN{Antonious M. Girgis\IEEEauthorrefmark{2},
        Amr El-Keyi\IEEEauthorrefmark{2}, Mohammed Nafie\IEEEauthorrefmark{2}\IEEEauthorrefmark{3} and
        Ramy Gohary\IEEEauthorrefmark{1}
        }
    \IEEEauthorblockA{\IEEEauthorrefmark{2}Wireless Intelligent Networks Center (WINC), Nile University, Cairo, Egypt\\
        \IEEEauthorrefmark{3}EECE Dept., Faculty of Engineering, Cairo University, Giza, Egypt\\
        \IEEEauthorrefmark{1}Department of Systems and Computer Engineering, Carleton University, Ottawa, ON, Canada\\
        Email: \{a.mamdouh@nu.edu.eg, \{aelkeyi, mnafie\}@nileuniversity.edu.eg, gohary@sce.carleton.ca\}
    }
}

\maketitle

\begin{abstract}

In this paper, proactive resource allocation based on user location for point-to-point communication over fading channels is introduced, whereby the source must transmit a packet when the user requests it within a deadline of a single time slot. We introduce a prediction model in which the source predicts the request  arrival $T_p$ slots ahead, where $T_p$ denotes the prediction window (PW) size. The source allocates energy to transmit some bits proactively for each time slot of the PW with the objective of reducing the transmission energy over the non-predictive case. The requests are predicted based on the user location utilizing the prior statistics about the user requests at each location. We also assume that the prediction is not perfect. We propose proactive scheduling policies to minimize the expected energy consumption required to transmit the requested packets under two different assumptions on the channel state information at the source. In the first scenario, offline scheduling, we assume the channel states are known a-priori at the source at the beginning of the PW. In the second scenario, online scheduling, it is assumed that the source has causal knowledge of the channel state. Numerical results are presented showing the gains achieved by using proactive scheduling policies compared with classical (reactive) networks. Simulation results also show that increasing the PW size leads to a significant reduction in the consumed transmission energy even with imperfect prediction.

\end{abstract}

\IEEEpeerreviewmaketitle

\section{Introduction}

The increasing number of wireless devices, e.g., smart phones, tablet computes, that are accessing wireless networks is leading to rapid evolution of the traffic load. In this context, wireless networks should be enhanced to support this increasing load under limited resources while satisfying the desired quality of service~\cite{misc}. One of the critical resources is the transmission energy, where minimizing energy consumption reduces the cost of downlink transmission and extends the battery life of wireless devices in uplink transmission. In this work, our objective is minimizing the transmission energy by utilizing the predictability of human behavior.

There are various studies, e.g., Song et al.~\cite{song2010limits} and Jensen~\cite{jensen2010estimating}, showing that user behavior is highly and precisely predictable. Recently, utilizing the predictability of user behavior has received considerable attention in many applications for communication networks. El Gamal in~\cite{el2010proactive} showed that communication systems can use lower bandwidth to achieve a required outage probability via predicting incoming requests. In~\cite{gungor2011proactive}, the authors  proposed a proactive source coding scheme for dynamic content that minimizes the communication cost, where the source non-causally knows the transmission cost before starting the transmission. In~\cite{tadrous2014can},~\cite{alotaibi2015game}, it has been shown that proactive data download can achieve a win-win situation, in which the users minimize their payments and the carrier maximizes its profit. De Mari in~\cite{de2014energy} introduced energy efficient scheduling of delay constrained transmission by assuming several scenarios of the future channel state knowledge.

 As a result of the regularity of user's trajectory on weekdays, future locations of the user usually can be predicted with high accuracy. In~\cite{song2010limits},~\cite{jensen2010estimating}, the mobile phone is used as a sensor to collect information about user behavior, enabling the network to predict future locations of the user. In~\cite{burbey2011predicting}, experiments for predicting future locations with high accuracy was introduced based on a Markov model. Therefore, predictive networks can track transitions in the user location, and  predict the future locations of the user. Predicting the user location can be very useful in predicting its requested traffic. For example, an application can run in the background of the smart devices to collect data about user requests from each location and submit this data to the network as proposed in~\cite{el2010proactive}. Hence, predictive networks can predict the user requests in future slots based on the user location at the current time slot (TS).

Optimal scheduling of delay-constrained transmission was investigated in~\cite{lee2009energy}, whereby the source has $T$ slots to transmit a packet of $B$ bits. These traditional networks are deemed \textit{reactive}, where the request is served after it is initiated by the user, and hence, the source is limited by the deadline of $T$ slots. Let us assume that the network can \textit{perfectly} predict the user request $T_p$ slots in advance, where $T_p$ is the prediction window (PW) size. Hence, the source can transmit a packet of $B$ bits within $T_p+T$ slots instead of $T$ slots only which can lead to minimizing the transmission energy.

In this paper, we develop a proactive scheduling policy that consumes an amount of energy lower than that consumed in reactive networks by proactively  transmitting some bits from the packet during the PW, i.e., before the user requests it. We assume that the deadline of the request is a single TS, i.e., $T=1$, and the prediction of the request arrival is \textbf{inaccurate} through the PW. Therefore, at some instants, the network might incorrectly predict the request arrival and transmit some bits proactively which leads to wasting the transmission energy. We propose a smart proactive scheduler that estimates the probability of the request arrival using the current user location and the prior statistics.  The proposed proactive scheduler minimizes transmission energy by selecting the number of bits transmitted proactively at each TS of the PW based on the estimated probability of requesting the packet. We propose two different scheduling polices, namely offline scheduling and online scheduling, where the prediction in both polices is not perfect. Under offline scheduling policy, the network is assumed to have non-causal knowledge of the channel states. This is not a practical assumption but we use it to gauge the performance bounds of our proposed algorithms. While under the online policy, the network is assumed to causally know the channel state information at the beginning of each TS. We compare the performance of predictive networks with reactive networks using numerical simulations to show the gain achieved by using the proposed proactive scheduling strategies.

The remainder of the paper is organized as follows. In Section~\ref{sys}, we define the system model and formulate the optimization problem. In Section~\ref{off}, we introduce the proactive scheduling strategy for the offline scheduling scenario. Section~\ref{online} presents the optimal scheduler for the online scheduling scenario. In Section~\ref{results}, we present our numerical results. Finally, we conclude the paper in Section~\ref{conclusion}.

\section{SYSTEM MODEL}
\label{sys}

Consider a source communicating with a single user in a time-slotted wireless network. The user can make a request for a packet of size $B$ bits at the beginning of any TS and each request has a deadline of a single TS. Therefore for reactive networks, when the user requests the packet at a certain TS, the source must transmit $B$ bits before the end of this TS. In predictive networks, the source has the capability to predict the request arrival $T_p$ slots ahead, where $T_p$ represents the PW size and $T_p \in \mathbb{N_{+}}$. Without loss of generality, we focus on the request initiated by the user at TS $t=1$, where $t$ denotes the index of slots in the time duration $T_p+1$ that consists of the PW time slots and the deadline, i.e., $t \in \lbrace T_p+1,\cdots,1\rbrace$, where the index of time $t$ is in descending order (similar to~\cite{lee2009energy}). Let $\mathcal{T}_p=\left\{T_p+1,\cdots,2\right\}$ denote the set of slots belonging to the PW and $t=1$ is the actual arrival time of the predicted request as shown in Fig.~\ref{fig1}. The binary  parameter $I$ is an indicator to the request arrival at TS $t=1$. It is defined as
\begin{equation}
I=\left\{\begin{array}{l l}
1 & \text{if the user requests the packet at TS $t=1$}\\
0 & \text{otherwise}.
\end{array}\right.
\end{equation}
We assume that the prediction is inaccurate. Hence, the source does not know the exact value of the indicator $I$ at each TS $t\in\mathcal{T}_p$. Thus, during the PW, the indicator $I$ is a Bernoulli random variable with parameter $p_t$ that denotes the probability of the request arrival at the last TS $t=1$.

We assume that the source is transmitting with the channel capacity under unit
variance white Gaussian Noise. Therefore, the number of bits transmitted proactively at TS $t$ is given by
\begin{equation}
b_t=WT\log_2\left(1+\frac{h_t E\left(b_t,h_t\right)}{WT}\right)
\end{equation}
where $E\left(b_t,h_t\right)$ is the amount of energy used for transmission at time slot $t$, and $h_{t}$ denotes the channel state at TS $t$. The channel states
$h_{t}$, $t=T_p+1,\cdots,1$, are assumed independent and identically distributed(i.i.d) random variables varying from one TS to another according to a known continuous distribution. Furthermore, $W$ and $T$ denote the channel bandwidth and the TS duration respectively. For simplicity, let $WT$ equal unity so that the energy consumed to transmit $b_t$ at TS $t$ is given by
\begin{equation} \label{eqn2}
E\left(b_t,h_t\right)=\frac{2^{b_t}-1}{h_t}.
\end{equation} Moreover, let $\beta_t$ denote the remaining bits of the packet at beginning of TS $t$, i.e., after transmitting some bits in earlier time slots, where $\beta_{T_p+1}=B$ and $\beta_t$ is given by
\begin{equation}
\begin{aligned}
\beta_t=B-\sum_{i=t+1}^{T_p+1}b_i=\beta_{t+1}-b_{t+1}.
\end{aligned}
\end{equation}

\begin{figure}[!t]
\centering
\includegraphics[width=8 cm, height=2.5 cm]{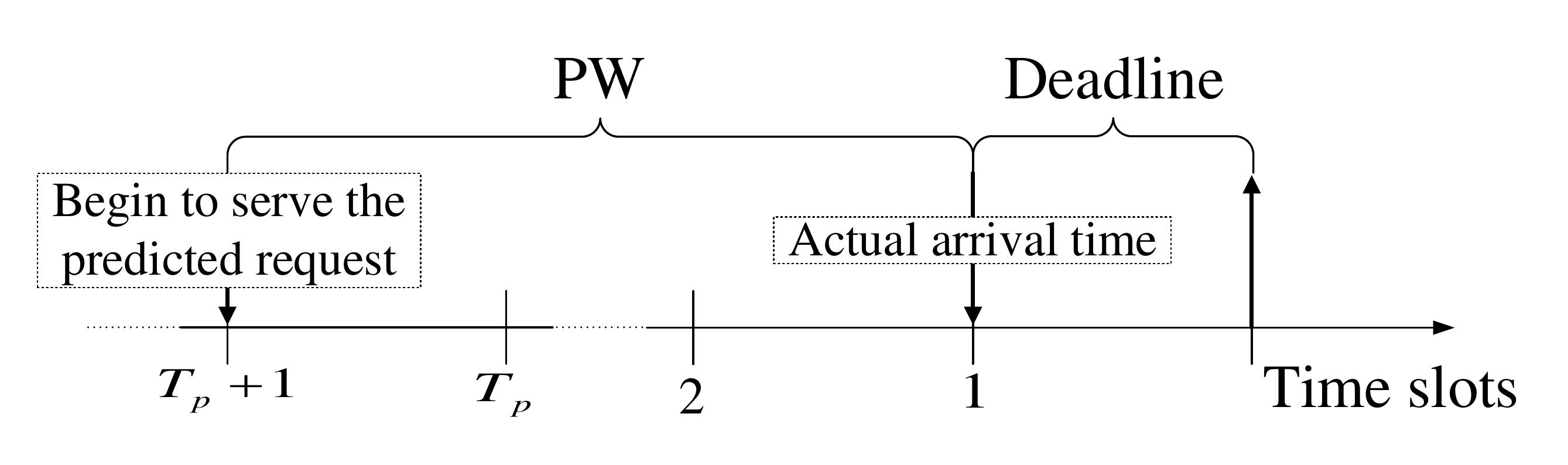}
\caption{Predictive Networks}
\label{fig1}
\end{figure}

\subsection{Problem formulation}

Our objective is to get an energy-efficient scheduler that minimizes the consumed energy  to deliver the packet before the end of the deadline upon request. The scheduler transmits proactively some bits during the PW before the user requests the packet at TS $t=1$. Hence, the source estimates the mean value of the indicator $I$ through the PW to determine the number of bits transmitted at each TS $t\in\mathcal{T}_p$\footnote{If we assume that the user requests are predicted perfectly, then the scheduling problem is equivalent to scheduling $B$ bits in $T_p+1$ time slots as in~\cite{lee2009energy}.}. Accordingly, the proactive scheduling policy can be obtained by solving the following optimization problem
\begin{subequations}\label{optimzation}
\begin{align}
\min_{b_{T_p+1},\cdots,b_{1}}&\mathbb{E}\left\{\sum_{t=2}^{T_P+1}E\left(b_t,h_t\right)+E\left(b_{1},h_{1}\right) I\right\} \label{obj} \\
\text{subject to} & \sum^{T_p+1}_{t=1}b_t=B \label{c1}\\
 & b_{t}\geq 0 \quad \forall \ t\in \lbrace 1,\cdots,T_p+1\rbrace
\end{align}
\end{subequations}
where $\mathbb{E}$ denotes the expectation operator.  The cost function~\eqref{obj} represents the total expected energy of transmission. It consists of two terms. The first term represents the total consumed energy during the PW to transmit some bits proactively, and the second term represents the consumed energy during the deadline to complete the transmission of the packet if and only if the user requests it at TS $t=1$, i.e., if $I=1$. The constraint~\eqref{c1} assures that no violation of the deadline occurs.

At each TS of the PW, $t\in \mathcal{T}_p$, there is a tradeoff between proactively transmitting more bits to minimize the energy consumption if the request arrives at $t=1$, and reducing the proactive transmission to reduce the wasted energy if the request does not arrive at $t=1$. Hence, the allocated bits $b_t$, $t\in\mathcal{T}_p$, depend heavily on the expected value of the indicator $I$ calculated at TS $t$, i.e., $ p_t$. In the last TS $t=1$, the remaining bits $b_1$ are transmitted using the energy $E\left(b_1,h_1\right)$ if and only if $I=1$.

In reactive networks, i.e., when $T_p=0$, if the user requests the packet at TS $t=1$ the source has a single TS to transmit it. Hence, the consumed energy $E_{R}$ is given by
\begin{equation}
E_{R}=\frac{2^{B}-1}{h_1}.
\end{equation}
\subsection{Dynamic estimation of the probability of request arrival}
\label{estimate}
Since the probability of the request arrival, $p_t$, is significantly correlated with the user location. We develop a dynamic estimate of $p_t$ through the PW based on the current user location. Let $X\left( t\right)$ denote the user location at TS $t$ that takes a value from the set $\textbf{S}$, where $\textbf{S}$ is a finite set of $k$ locations, $\lbrace l_1,l_2,\cdots,l_k\rbrace$, containing all possible locations of the user. We assume that  $X \left(T_p+1\right),\cdots,X \left(1\right)$ is a stationary Markov process~\cite{burbey2011predicting} where $\textbf{S}$ represents the state space of the process and the $k\times k$ transition probability matrix is denoted by $\mathbf{L}$ where the $(i,j)$th element of $\mathbf{L}$ is given by $L_{i,j}\!=\!\text{Pr}\lbrace \left.X\left(t\right)=l_j\right|X\left(t+1\right)=l_i\rbrace$. Note that the transition probabilities can be estimated by observing the mobility pattern of the user over a sufficiently long time interval. Let $g_i$ denote the probability that the user requests the packet given that the user is located at location $l_i$, i.e., $g_i=\text{Pr}\lbrace \left. I=1\right|X\left(1\right)=l_i\rbrace$. Thus, the row vector $\textbf{g}=\left[g_1,\cdots,g_k\right]$ denotes the prior statistics vector. Note that each element $g_i$ of the statistics vector $\textbf{g}$ can be computed by counting the number of times that the user requests the packet from the location $l_i$.

Once the source observes $X\left(t\right)$, $t\in\mathcal{T}_p$, the source can estimate the probability of request arrival at TS $t=1$  as follows
\begin{equation} \label{eqn1}
\begin{aligned}
p_t&=\text{Pr}\lbrace \left.I=1\right|X\left(t\right)\rbrace\\
&=\sum_{i=1}^{k} \text{Pr}\lbrace \left.X\left(1\right)=l_i\right|X\left(t\right)\rbrace\ g_i\\
&=\pmb{\pi}_t\ \textbf{L}^{t-1}\ \textbf{g}^{T}\\
\end{aligned}
\end{equation}
where $\pmb{\pi}_t=\left[\pi_1,\cdots,\pi_k\right]$ is the
observation row vector whose $j\text{-th}$ element $\pi_j=1$ if
$X\left(t\right)=l_j$ and the other elements $\pi_i=0
\quad \forall\ i\neq j$. The elements of the matrix
$\textbf{L}^{t-1}$ are the conditional probabilities of the user
location at TS $t=1$ given his location at $t$, and
$(\cdot)^T$ denotes the transpose operator. Note that the
probability of the request arrival $\lbrace
p_t\rbrace_{t=2}^{T_p+1}$ changes from one TS to another during
the PW due to the change in the user location.

In the next sections, we propose an energy-efficient scheduler
that can minimize the expected energy of transmission compared to
reactive networks.

\section{OFFLINE SCHEDULING} \label{off}
Here, we assume that the source
non-causally knows the channel states $\lbrace h_{t}\rbrace^{T_p+1}_{1}$ at the first TS of the PW, i.e., the source has the ability to predict the future channel states~\cite{duel2007fading}. Although this assumption is non-realistic, the results of this section are useful in providing an upper bound on the performance of causal scheduling algorithms presented in Section~\ref{online}. Moreover, the offline scheduling is used to get suboptimal solution for online scheduling when $T_p>1$ as we will discuss later in Section~\ref{online1}. First, we get a closed-form expression of the optimal scheduler when $T_p=1$. Next, we discuss the general case for $T_p>1$.

\subsection{Proactive scheduling policy for $T_p=1$}
\label{offtp1} At TS $t=2$, the source transmits $b_2$ bits based on its information about $h_1$, $h_2$ and $p_2$, where $p_2$ is obtained from equation~\eqref{eqn1}. At TS $t=1$,
the source takes one of two actions: a) If the user requests the
packet, i.e., the prediction is correct, the remaining bits of
the packet $b_1=B-b_2$ are transmitted, b) In the case of erroneous
prediction, i.e., the user does not request the packet, the source
will not transmit $b_1$, and the source loses the amount of energy consumed at TS $2$.
\begin{p}
\label{Proposition_1}
The optimum number of bits for proactive transmission when $T_p=1$
is given by
\begin{equation}\label{P1}
\begin{aligned}
b_{2}\ =&\ \left<  \frac{B}{2} + \frac{1}{2}\log_{2}\left(\frac{h_{2}}{\overline{h_1}}\right) \right>_{0}^{B}\\
\end{aligned}
\end{equation}
where  $\overline{h_1}\ =\ \frac{h_1}{p_2}$  and $\left<.\right>_{0}^{B}$ denotes the truncation from below at $0$ and from above at $B$.
\end{p}

\begin{proof}
See Appendix
\end{proof}
We can see from~\eqref{P1} that the parameter $p_2$ has significant effect on determining the number of bits allocated to TS $2$. Decreasing $p_2$, i.e., decreasing the probability that the packet would be requested at $t=1$, leads to increasing the effective channel gain $\overline{h_1}$ which reduces the number of bits transmitted proactively ($b_2$). In other words $p_2\rightarrow0$ leads to $\overline{h_1}\rightarrow \infty$ and $b_2\rightarrow0$. On the other hand, if $p_2=1$, i.e., the source knows certainly that $I=1$, the problem reduces to the scheduling problem introduced in~\cite{lee2009energy}. In this case, when $h_2=\overline{h_1}$, the packet is divided equally between the two time slots. However, when $h_2>\overline{h_1}$ more than half the packet is transmitted proactively due to favorable channel state at TS $t=2$. We note that the source does not transmit some bits from the packet proactively if $h_2 < 2^{-B}\overline{h_1}$. Hence, the channel state $h_2$ must be greater than $2^{-B}\overline{h_1}$ to allow proactive scheduling.

\subsection{Proactive scheduling policy for $T_p>1$}
\label{offtp} In this case the source has $T_p+1$ time slots to
serve the predicted request. Therefore, the source allocates a
number of bits $b_t$ to be transmitted at TS $t$ 
according to the current channel state $h_t$ compared to the
future channel states $\lbrace h_i\rbrace_{i=1}^{t-1}$ and the
expected value of the indicator $I$ at this TS, i.e., $p_t$.
\begin{p}
The optimal scheduling policy for $T_p>1$ is given by
\begin{equation} \label{eqn4}
\begin{aligned}
&b_t=\left<\log_2\left(\frac{h_t}{\varepsilon^t_{\text{th}}}\right)\right>^{\beta_t}_{0}\\
\text{where}\quad&\textbf{H}_t=\lbrace h_{t},h_{t-1},\cdots,\overline{h_1}\rbrace, \quad\overline{h_1}=\frac{h_1}{p_t}\\
&\bar{\textbf{H}}_t=\lbrace \left.h\right|h \in \textbf{H}_t , h > \varepsilon^t_{\text{th}}\rbrace\\
&N=\vert\bar{\textbf{H}}_t\vert\\
&\varepsilon^t_{\text{th}}=2^{-\frac{\beta_t}{N}}\mathbb{G}
\left(\bar{\textbf{H}}_t\right).
\end{aligned}
\end{equation}
\end{p}
The operation $\vert\textbf{H}\vert$ denotes the number of elements in the set $\textbf{H}$ and $\mathbb{G} \left(\textbf{H}\right)$ denotes the geometric mean of the elements in this set. Notice that the set $\textbf{H}_t$ contains
the present channel state and the channel states of the future time slots. However, the set $\bar{\textbf{H}}_t$ is a subset of $\textbf{H}_t$ containing the channel states that are greater than threshold value $\varepsilon^t_{\text{th}}$.

\begin{proof}
At TS $t$, the optimization problem~\eqref{optimzation} can be formulated as the following optimization problem
\begin{equation} \label{eqn3}
\begin{aligned}
&\min_{b_t,\cdots,b_1}& &\sum_{i=2}^{t}\frac{2^{b_i}-1}{h_i}+ \left(\frac{2^{b_1}-1}{h_1}\right) p_t\\
&\ \text{subject to}& &\sum_{i=1}^{t} b_{i} \ = \ \beta_{t}\\
&& &b_{i}\ \geq \ 0 \quad \forall \ i\in \lbrace t,t-1,\cdots,1\rbrace.
\end{aligned}
\end{equation}
The objective function is the expected energy of transmission calculated at the TS $t$. The Hessian of the objective function of problem~\eqref{eqn3} is a diagonal matrix with non-negative diagonal elements, and hence, the optimization problem is a convex problem. The optimal solution can be obtained by solving the Karush-Kuhn Tucker (KKT) conditions yielding (\ref{eqn4}). The details are omitted due to space considerations.
\end{proof}

At TS $t$,  the scheduler solves the problem~\eqref{eqn3} to get the vector $\mathbf{b}=\left[b_t,\cdots,b_1\right]$. However, only the value of $b_t$ is utilized where the scheduler transmits $b_t$ bits at TS $t$. In the next time slot, $t-1,$ the scheduler resolves  problem~\eqref{eqn3} again after updating the probability of requesting the packet based on the current user location. Note that the $t$th TS is utilized for proactive transmission if and only if $h_t > \varepsilon^t_{\text{th}}$, where the channel state $h_1$ is replaced by $\overline{h_1}$ that reflects the effect of the accuracy of  prediction on the proactive scheduler. Decreasing $p_t$ leads to increasing $\overline{h_1}$ and $\varepsilon^t_{\text{th}}$ which reduces the number of bits transmitted proactively. Furthermore, increasing the PW size $T_p$ gives the source a better chance to select the favorable time slots that have $h_t > \varepsilon^t_{\text{th}}$ to transmit more parts of the packet proactively which decreases the total consumed energy.

\section{ONLINE SCHEDULING}
\label{online} In this section, we assume that the source has
causal channel state information, i.e., at the beginning of TS $t$
the source knows $h_t$, but future channel states $\lbrace h_i\rbrace^{t-1}_{i=1}$ are not known. Let $A_t$ denote the information available to the source at TS $t$, i.e.,
\begin{equation}
A_t=\left\{\begin{array}{ll}
\left(\beta_t,h_t,p_t\right) & t\in\mathcal{T}_p\\
\left(\beta_t,h_t,I\right) & t=1
\end{array}
\right. .
\end{equation}
We aim to get the proactive scheduling policy $\mathcal{G}^{*}=\left[b_{T_p+1}\left(A_{T_p+1}\right),\cdots,b_1\left(A_1\right)\right]$ which is a sequence of functions mapping the available information to the source at TS $t$ into a number of bits transmitted at this TS to minimize the cost function~\eqref{obj}, where the expectation with respect to the random indicator $I$ and the channel states $h_{T_p+1},\cdots,h_1$. The optimal policy $\mathcal{G}^{*}$ can be obtained by solving the problem~\eqref{optimzation} recursively using the standard dynamic programming algorithm
\begin{subequations}  \label{belman}
\begin{align}
J_1\left(\beta_1,I,h_1\right)&=E\left(\beta_1,h_1\right)I&\label{belman:1}\\
\label{belman:2}
J_2\left(\beta_2,p_2,h_2\right)&=\min_{0\leq b_2\leq \beta_2}  E\left(b_2,h_2\right)+\mathbb{E}_{I}\left\{\overline{J}_{1}\left(\beta_2-b_2,I\right)\right\}  &\\\label{belman:3}
J_t\left(\beta_t,p_t,h_t\right)&=\min_{0\leq b_t\leq \beta_t} E\left(b_t,h_t\right)+\overline{J}_{t-1}\left(\beta_t-b_t,p_t\right) 
\end{align}
\end{subequations}
where $\overline{J}_{t-1}\left(\beta,p\right)=\mathbb{E}_{h}\left\{J_{t-1}\left(\beta,p,h\right)\right\}$, represents the minimum expected energy to transmit $\beta$ bits through the remaining $t-1$ slots given $\text{Pr}\left(I=1\right)=p$. We take the expectation to the random indicator $I$ when applying the backward recursion in~\eqref{belman:2} since the source does not know the value of $I$ in the PW time slots. In~\eqref{belman:1}, the source transmits all the remaining bits at the last TS $t=1$ if and only if $I=1$ to complete the transmission of the packet before the end of the deadline. In~\eqref{belman:2}, and~\eqref{belman:3}, the source determines the optimal bit $b_t^{*}$ allocated to the TS $t\in\mathcal{T}_p$ that minimizes the current consumed energy plus the expected energy to transmit the remaining bits through the remaining slots, where the expectation is taken with respect to the future channel states and the random indicator $I$. The optimality of the algorithm is verified according to Bellman's equations~\cite{bertsekas1995dynamic}.

\subsection{Proactive scheduling policy for $T_p=1$}
\begin{p} The optimum number of bits allocated to the TS $2$ is given by
\begin{equation}
\begin{aligned}
b_2&=\left<\frac{B}{2}+\frac{1}{2}\log_2\left(h_2 \nu_1 p_2\right)\right>^{B}_{0}\\
&\text{where}\quad\nu_1=\mathbb{E}_h\left\{\frac{1}{h}\right\}.
\end{aligned}
\end{equation}
\end{p}
\begin{proof}
By applying Bellman's equations~\eqref{belman}, we have $\mathbb{E}_{I}\left\{\overline{J}_1\left(\beta,I\right)\right\}=\left(2^{\beta}-1\right)\nu_1 p_2$, and hence, the solution to the optimization problem at TS $2$ is given by
\begin{equation} \label{eqn10}
b_2=\arg \min_{0\leq b_2 \leq B}\ \frac{2^{b_2}-1}{h_2}+\left(2^{B-b_2}-1\right)\nu_1p_2
\end{equation}
Since the problem~\eqref{eqn10} is a convex problem, the optimal $b_2$ is obtained by equating the first derivative to zero.
\end{proof}

\subsection{Proactive scheduling policy for $T_p>1$}\label{online1}

There is no closed form expression for the optimal scheduling policy for $T_p>1$ since the closed form expression for the cost function $\overline{J}_t\left( \beta,p\right)$ defined in~\eqref{belman} does not exist. Therefore the optimal scheduler is obtained numerically by using the discretization method (see Section~$6.4$ in~\cite{bertsekas1995dynamic}). However, we can get two suboptimal solutions similar to~\cite{lee2009energy} whose performance is close to the optimal one.

\subsubsection{Certainty equivalent Scheduler (CES)}

We obtain the certainty equivalent scheduler (CES) by applying the following steps at each TS $t\in\mathcal{T}_p$: a) the source replaces  each uncertain variable at TS $t$ with its mean. Hence, we assume the channel inversion of the future time slots $\lbrace \frac{1}{h_i}=\nu_1\rbrace^{t-1}_{i=1}$ and $\mathbb{E}\left\{I\right\}=p_t$. b) the source determines $b_t$ by applying the offline scheduling described in Section~\ref{offtp} over the following channels inversion: $\frac{1}{h_t}$, $\lbrace \frac{1}{h_i}=\nu_1\rbrace^{t-1}_{i=1}$. Therefore, the CES is given by

\begin{equation}
\begin{aligned}
&b_t=\left<\frac{\beta_t}{t}+\frac{t-1}{t}\log_2\left(\frac{h_t}{\varepsilon^t_{\text{CES}}} \right)\right>^{\beta_t}_{0}\\
\text{where}\quad&\varepsilon^t_{\text{CES}}=\frac{1}{\nu_1 p^{\frac{1}{t-1}}_t}
\end{aligned}
\end{equation}

\subsubsection{Suboptimal II Scheduler}

The suboptimal II scheduler is obtained by relaxing the constraint $0\leq b_t\leq\beta_t$. Therefore, we can get an approximate closed-form expression for $\overline{J}_{t-1}\left(\beta,p\right)$ by using mathematical induction as
\begin{equation}
\begin{aligned}
\tilde{\overline{J}}_{t-1}\left(\beta,p\right)=&\left(t-1\right) 2^{\frac{\beta}{t-1}}\mathbb{G}\left(\nu_{t-1},\cdots,\nu_1\right)p^{\frac{1}{t-1}}\\
& -\left(t-2+p\right)\nu_1\\
\end{aligned}
\end{equation}
where $\nu_i=\mathbb{E}_h\left\{\frac{1}{h^{\frac{1}{i}}}\right\}^i$ and the operation $\mathbb{G}\left(\nu_{t},\cdots,\nu_1\right)$ denotes the geometric mean $\left(\prod^{t}_{i=1}\nu_i\right)^{\frac{1}{t}}$. Therefore, at the beginning of each TS $t\in\mathcal{T}_p$, the number of bits $\left(b_t\right)$ is determined by solving the following optimization problem
\begin{equation}\label{eqn13}
b_t=\arg \min_{b_t}\ \frac{2^{b_t}-1}{h_t}+\tilde{\overline{J}}_{t-1}\left(\beta_t-b_t\right).
\end{equation}

Note that problem~\eqref{eqn13} is a convex optimization problem. Therefore, we get the following suboptimal II scheduler by equating the first derivative to zero and using the truncation operator to maintain the constraint $0 \leq b_t \leq \beta_t$ which yields
\begin{equation}\label{eqn12}
\begin{aligned}
& b_t=\left<\frac{\beta_t}{t}+\frac{t-1}{t}\log_2\left(\frac{h_t}{\varepsilon^t_{\text{SubII}}}\right)\right>^{\beta_t}_{0}\\
\text{where}\quad&\varepsilon^t_{\text{SubII}}=\frac{1}{\mathbb{G}\left(\nu_{t-1},\cdots,\nu_1\right)p^{\frac{1}{t-1}}_t}.
\end{aligned}
\end{equation}
\section{NUMERICAL RESULTS}
\label{results}

In this section, we compare the performance of predictive networks and reactive networks to illustrate the gain that can be achieved by using the proposed proactive scheduling policies. We run $10^3$ Monte-Carlo simulations to evaluate the expected transmission energy of predictive and reactive networks. We assume that the user can exist in three locations, i.e., $k=3$. In each simulation, the elements of the transition matrix $\textbf{L}$ are generated with uniform distribution and each row is normalized. In addition, the elements of the vector $\textbf{g}$ are randomly generated according to a uniform distribution. At each simulation, the channel states $\lbrace h_t\rbrace^{T_p+1}_{t=1}$ are randomly generated according to a truncated exponential distribution with parameter $\lambda=1$ and threshold $t_o=0.001$. Then, we generate a sample path of the user location according to the stationary Markov chain described in Section~\ref{estimate}, where the initial location of the user $X\left(T_p+1\right)$ is generated according to the steady-state probability distribution of the Markov chain. Finally, the request initiated at TS $t=1$ is generated according to the user location $X\left(1\right)$ and the statistics vector $\textbf{g}$. For example, if $X(1)=l_2$, the indicator $I$ is generated as a Bernoulli random variable with parameter $g_2$.

Numerical results for the offline scheduling algorithm described in Section~\ref{offtp} are shown in Fig.~\ref{fig3} that displays the expected transmission energy versus the packet size for different values of $T_p$ ($T_p=0$ represents the reactive case). Fig~\ref{fig3} shows the gap between the expected energy of reactive and predictive networks. It is clear that increasing the PW size $T_p$ leads to decreasing the expected transmission energy.

\begin{figure}
\centering
\includegraphics[width=8 cm,height=6 cm]{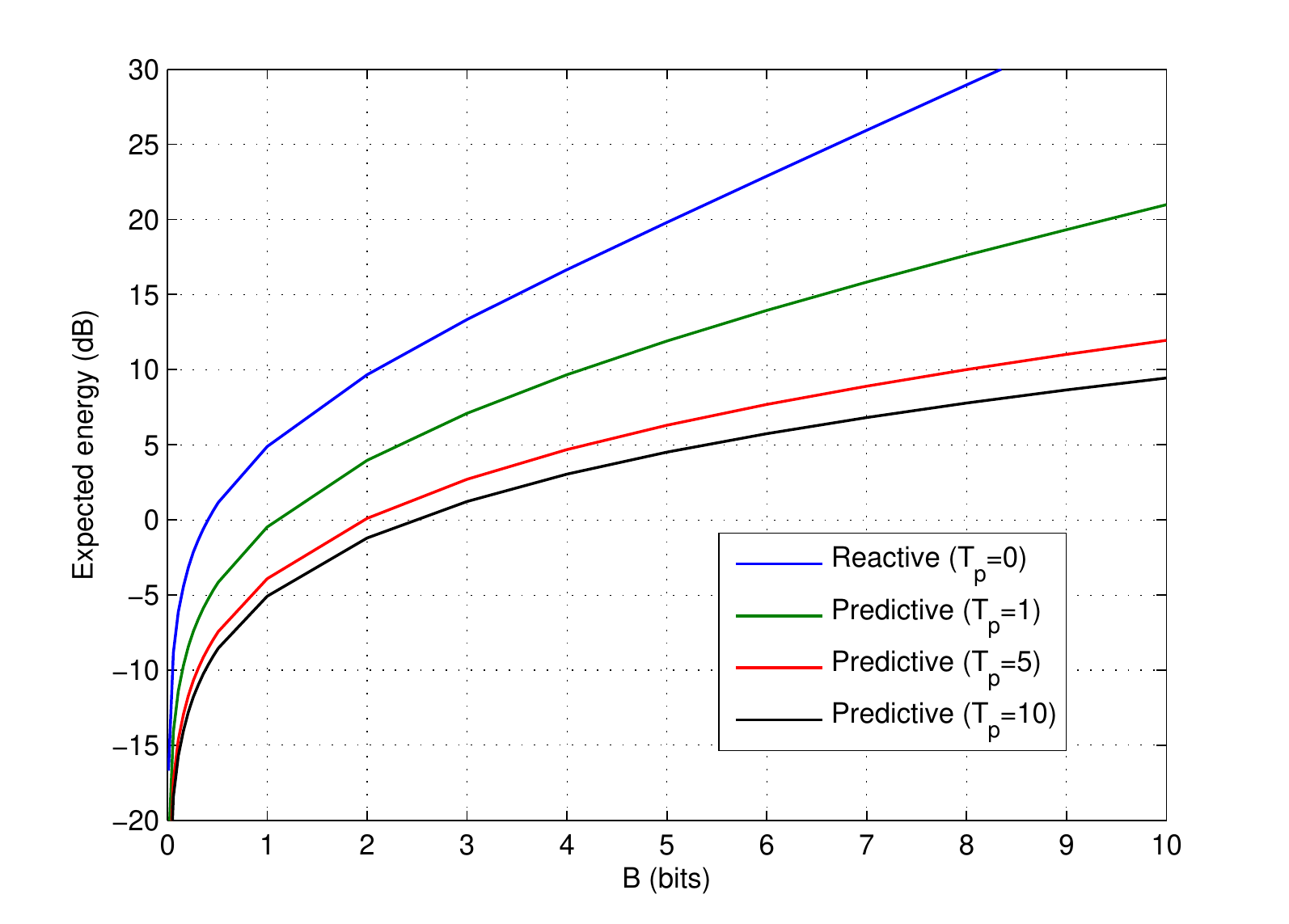}
\caption{Comparison between proactive scheduling schemes with different PW size and the reactive scheme.}
\label{fig3}
\end{figure}

In Fig~\ref{fig4}, we investigate the performance of the online scheduling schemes with PW size $T_p=4$ versus reactive networks. It is shown that the gap between the proposed algorithms and the reactive case increases when the packet size increases. We can also see that when the packet size increases, the suboptimal solutions converge to the optimal one obtained by using the discretization method. However, Suboptimal II performs better than CES since the suboptimal II is obtained by relaxing the constraint $0\leq b_t\leq \beta_t$ that might be satisfied for large packet size.

\begin{figure}
\centering
\includegraphics[width=8 cm,height= 6 cm]{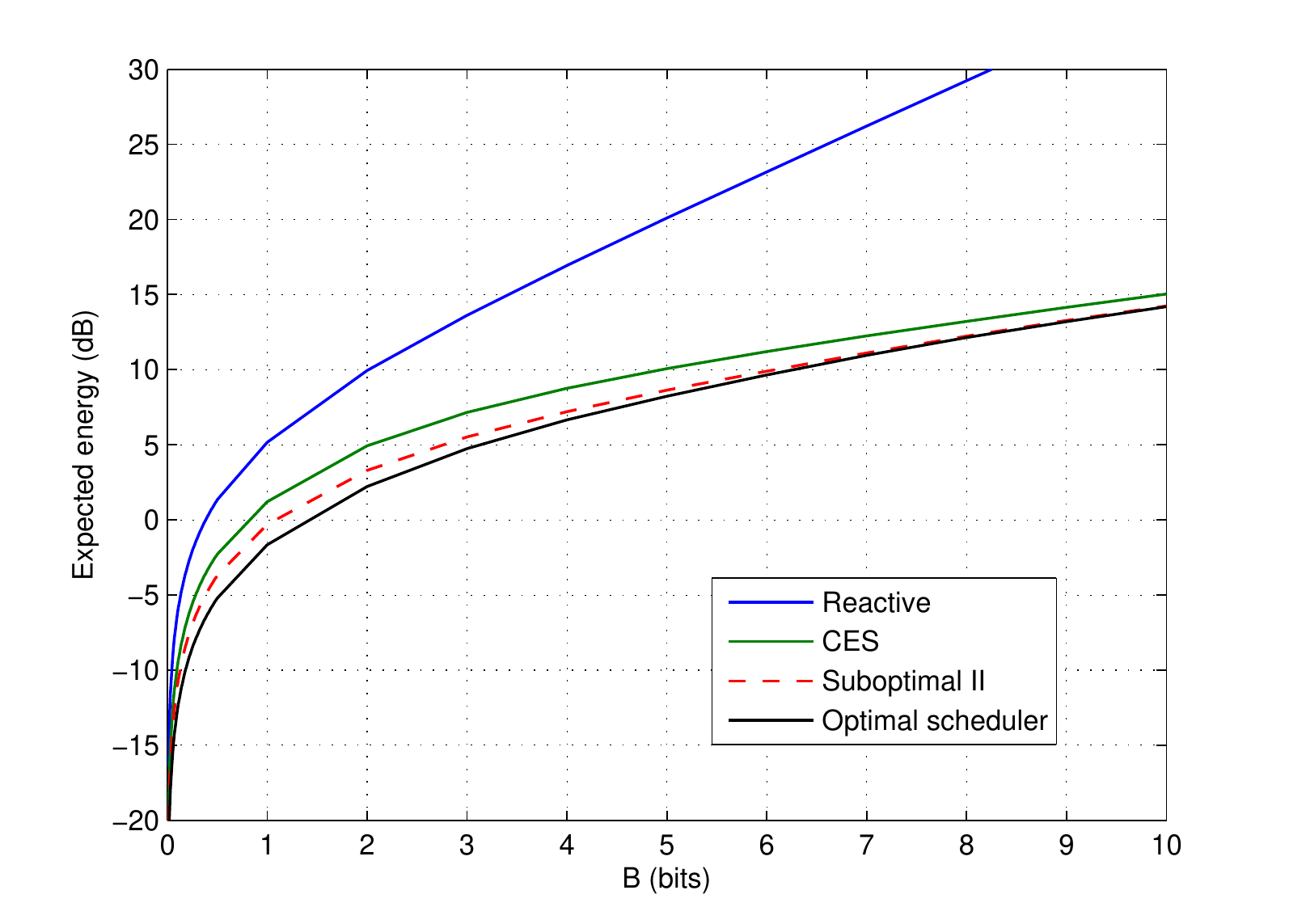}
\caption{Comparison of online scheduling schemes with PW size $T_p=4$ and reactive (traditional) scheme.}
\label{fig4}
\end{figure}
In Fig~\ref{fig2}, we display the amount of energy that can be saved by using our proactive scheduling strategy when PW size $T_p=1$ described in Section~\ref{offtp1}. The expected saved energy is defined by the difference between the expected transmission energy of traditional (reactive) and predictive
networks in dB units. The following values of $p_2$ that denote the probability of the request arrival at TS $t=1$, are simulated $p_2=\lbrace0.1,0.5,1\rbrace$ to show the impact of the prediction accuracy on the expected saved energy. The results show that proactive scheduling can achieve gain for each value of $p_2$, where the expected saved energy increases as the packet size $B$ increases.

\section{CONCLUSION}
\label{conclusion}

In this paper, we have proposed proactive scheduling schemes for delay-constrained traffic under two different assumptions on the channel state information at the source. Our schemes are based on predicting the request arrival $T_p$ slots ahead so that the source can optimize the transmission energy in each TS of the PW  with the objective of minimizing the expected transmission energy. Throughout this paper, we have assumed that prediction is not perfect and the effects of prediction errors on the proactive energy allocation were taken into account. Numerical results have been provided to demonstrate the superiority of the proactive scheduling strategies, where predictive networks achieve significant reduction in the transmission energy compared to reactive ones.

\appendix\label{Appendix_1}
The optimum scheduler is the solution to the problem~\eqref{optimzation} by taking the expectation with respect to the indicator $I$. Thus, the problem~\eqref{optimzation} can be reformulated as follows
\begin{equation} \label{l1}
\begin{aligned}
&\arg\min_{0\leq b_2\leq B}& &\frac{2^{b_2}-1}{h_2}+\left(\frac{2^{B-b_2}-1}{h_1}\right)p_2
\end{aligned}
\end{equation}
We define $f\left(b_2\right)$ as the objective function in~\eqref{l1}
\begin{equation}
f\left(b_2\right)=\frac{2^{b_2}-1}{h_2}+\left(\frac{2^{B-b_2}-1}{h_1}\right)p_2
\end{equation}
The first and second derivatives of $f\left(b_2\right)$ are given
by
\begin{eqnarray}
\frac{d f}{d b_2}&\!\!=\!\!&
\frac{2^{b_2}}{h_2}\ln\left(2\right)-\frac{2^{B-b_2}}{h_1}p_2\ln\left(2\right)\label{d1}\\
 \frac{d^2 f}{d
b^{2}_{2}}&\!\!=\!\!&\frac{2^{b_2}}{h_2}\left(\ln\left(2\right)\right)^2+\frac{2^{B-b_2}}{h_1}p_2\left(\ln\left(2\right)\right)^2
\label{d2}
\end{eqnarray}
where $\frac{d^2 f}{d b^{2}_{2}} \geq 0$ $\forall\ b_2 \geq 0$,
i.e., the optimization problem~\eqref{l1} is a convex problem.
Thus the optimum solution $b_2$ is obtained by setting the
first derivative to zero in~\eqref{d1} where the constraint $0\leq b_2\leq B$ is not violated by using the truncation
operator.

\begin{figure} [t]
\centering
\includegraphics[width=8 cm,height= 6 cm]{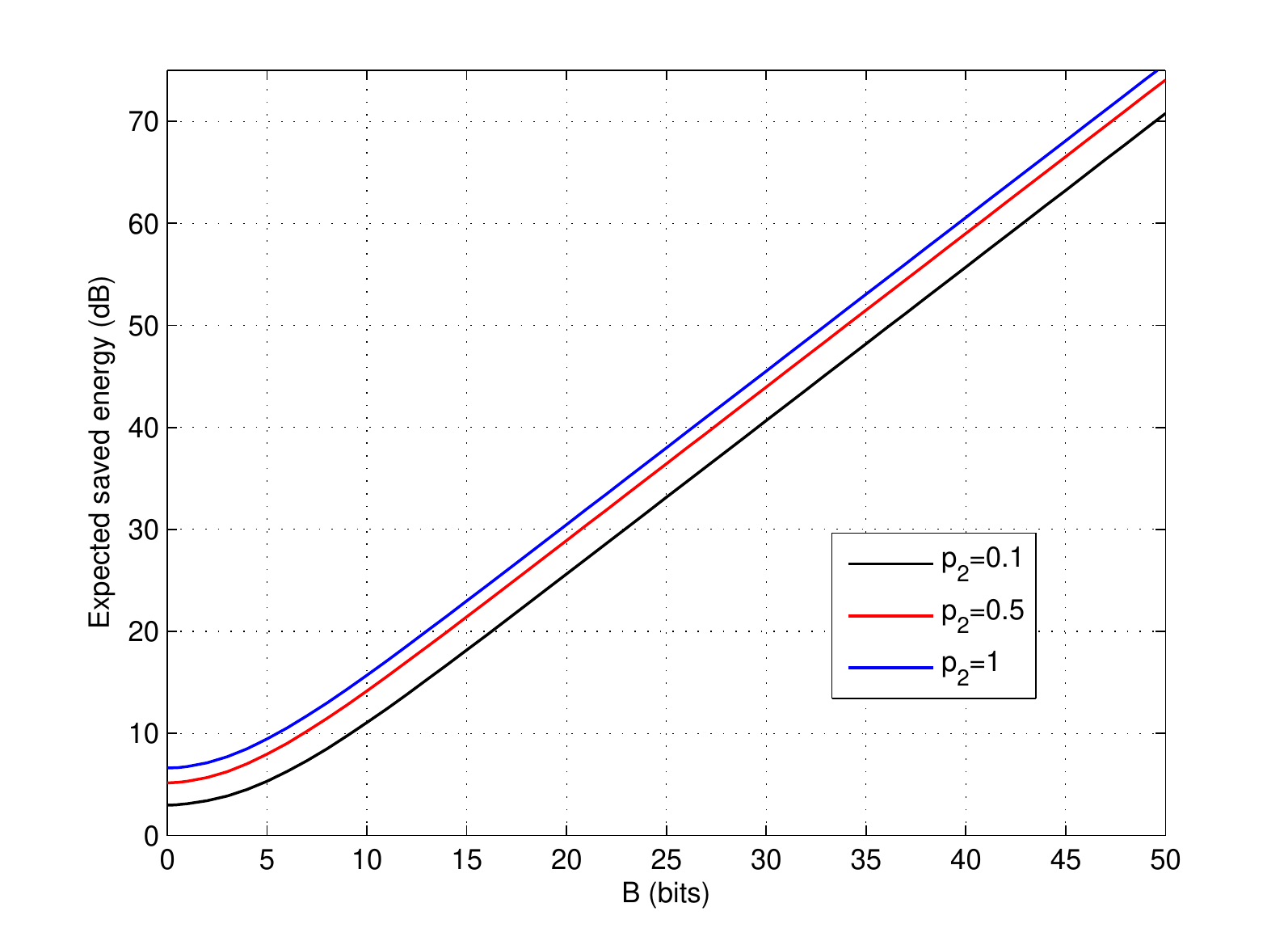}
\caption{Expected saved energy in dB for $T_p=1$.} \label{fig2}
\end{figure}

\nocite{*}
\bibliographystyle{IEEEtran}
\bibliography{IEEEabrv,mybibfile}

\end{document}